\documentclass{article}
\usepackage[a4paper,margin=2.54cm,left=2.54cm,right=2.54cm]{geometry}

\usepackage{amsfonts}
\usepackage{amsmath}
\usepackage{amssymb}
\usepackage{amsthm}

\usepackage{graphicx}
\usepackage{newtxtext}
\usepackage{newtxmath}
\usepackage{natbib}
\usepackage{hyperref}
\hypersetup{
    colorlinks = true,
    urlcolor   = blue,
    citecolor  = black,
}
\newtheorem{lemma}{Lemma}

\newcommand{\RomanNumeralCaps}[1]

\usepackage{bm, color}
\def\e{e}
\usepackage{authblk}

\def\mathsfbi{}
\def\mathsfi{}

\newcommand\changes[1]{{\textcolor{black}{#1}}}

\title{Bayesian comparison of stochastic models of dispersion}

\author{Martin T. Brolly}
\author{James R. Maddison}
\author{Aretha L. Teckentrup}
\author{Jacques Vanneste}
\affil{School of Mathematics and Maxwell Institute for Mathematical Sciences, University of Edinburgh,
King’s Buildings, Edinburgh EH9 3FD, UK}

\begin{document}
\maketitle

\begin{abstract}

Stochastic models of varying complexity have been proposed to describe the dispersion of particles in turbulent flows, from simple Brownian motion to complex temporally and spatially correlated models. A method is needed to compare competing models, accounting for the difficulty in estimating the additional parameters that more complex models typically introduce.
We employ a data-driven method, Bayesian model comparison (BMC), which assigns probabilities to competing models based on their ability to explain observed data.
We focus on the comparison between the Brownian and Langevin dynamics for particles in two-dimensional isotropic turbulence, with data that consists of sequences of particle positions obtained from simulated Lagrangian trajectories. We show that, while on sufficiently large timescales the models are indistinguishable, there is a range of timescales on which the Langevin model outperforms the Brownian model. While our set-up is highly idealised, the methodology developed is applicable to more complex flows and models of particle dynamics.
\end{abstract}

\section{Introduction}

Since Taylor introduced the notion of turbulent diffusion in the 1920s~\citep{Taylor}, a wide variety of stochastic models have been proposed to represent the dynamics of particles in turbulent flows \citep[e.g.][]{Thomson, Rodean, Majda_Kramer,Berloffii}. 
The Brownian dynamics used by Taylor models Lagrangian velocities as white noise processes and is a good approximation only on sufficiently long time scales. More complex models incorporate temporal and/or spatial correlation~\citep[e.g.][]{Griffa1996, Pasquero, lilly2017fractional}.
For example, Langevin dynamics incorporate autocorrelation in Lagrangian velocities by representing them as Ornstein--Uhlenbeck processes~\citep{OU}. It is in general unclear when such additional complexity leads to improved predictions rather than to overfitting. Given the increased difficulty and cost of implementing more complex models, a method for comparing  the performance of competing stochastic models for particle dynamics is needed.

To this end, we propose a data-driven approach: we apply \textit{Bayesian model comparison} (BMC)~\citep{jaynes_2003, KassRaftery, MacKay}, which assigns probabilities to competing models based on their ability to explain observed data. We focus on the comparison between the Brownian and Langevin models for particles in two-dimensional homogeneous isotropic turbulence, with  data that consists of sequences of particle positions obtained from simulated Lagrangian trajectories. While this set-up is highly idealised, the methodology developed is applicable to more complex flows and models of particle dynamics.

 Model comparison is complicated by two issues: (i) proposed models typically contain a number of parameters whose values are uncertain, and
 (ii) a measure of model suitability is required, balancing accuracy and complexity.
 The natural language for this problem is then that of decision theory (see e.g.\ \cite{BernardoSmith} and \cite{Robert} for an overview of decision problems under uncertainty); however, several philosophical issues therein, \changes{such as the choice of utility function and its subjectivity, }can be avoided by adopting the ready-made approach of BMC.
 BMC and the related technique Bayesian model averaging are gaining popularity in many applied fields~\citep{MarkNature, Min2007BMA, Carson, Mann2011_animals}.
 In this paper, we demonstrate the potential of BMC by comparing the Brownian and Langevin models of dispersion in two-dimensional turbulence. This provides a simple illustration of the BMC methodology while addressing a  problem of interest:  dispersion in two-dimensional turbulence has received much attention as a paradigm for transport and mixing in stratified, planetary-scale geophysical flows~\citep{Provenzale1995}, and
can be modelled with stochastic processes~\citep[e.g.][]{Pasquero, lilly2017fractional}.

The paper is structured as follows. We introduce the Brownian and Langevin models in \S\ref{sec-2} and review the BMC method in \S\ref{Methods}. In \S\ref{Results} we show how this method can be applied to discrete particle trajectory data; we also show results of a test case, where the data are generated by the Langevin model itself. In \S\ref{NS2D} we apply BMC to data from direct numerical simulations of two-dimensional turbulence. In \S\ref{Conclusions} we give our conclusions on the method.

\section{Models and data}\label{sec-2}

\subsection{Brownian and Langevin models}\label{Models}

The models of interest are the \textit{Brownian model}, which for passive particles in homogeneous and isotropic turbulence is given by
\begin{align}\label{Brownian}
    \text{d}\bm{X} = \sqrt{2\kappa}\,\text{d}\bm{W},
\end{align}
with $\kappa>0$, and the \textit{Langevin model}, which, under the same conditions, is given by
\begin{subequations}\label{Langevin}
\begin{align}
\text{d}\bm{X} &= \bm{U}\,\text{d}t,\\
\text{d}\bm{U} &= -\gamma \bm{U}\, \text{d}t+\gamma\sqrt{2k}\,\text{d}\bm{W},
\end{align}
\end{subequations}
with $\gamma, k > 0$, and where, in both cases, $\bm{W}$ is a vector composed of independent Brownian motions. We denote the models by $\mathcal{M}_B(\kappa)$ and $\mathcal{M}_L(\gamma,k)$.

 We note  some important characteristics of the two models. The Brownian model involves particle position, $\bm{X}$, as its only component, which evolves as a scaled $d$-dimensional Brownian motion, where $d$ is the number of spatial dimensions. This implies that particle velocity evolves as a white noise process. The model has one parameter, the diffusivity $\kappa$.
The validity of~\eqref{Brownian} is typically justified by arguments involving strong assumptions of scale separation between mean flows and small-scale fluctuations which rarely hold in applications~\citep{Majda_Kramer, Berloffii}.

The Langevin model, by contrast, involves two components, particle position and particle velocity, $(\bm{X},\bm{U})$. The velocity component evolves according to a mean-zero Ornstein--Uhlenbeck process, and position results from time integration of this velocity. The model has two parameters, $\gamma$ and $k$, where $\gamma^{-1}$ is a Lagrangian velocity decorrelation time and $k$ characterises the strength of Gaussian velocity fluctuations. The Brownian and Langevin models are the first two members of a hierarchy of Markovian models involving an increasing number of time derivatives of the position~\citep{Berloffii}. 

In practice, the Brownian model is favoured over the Langevin model for its simplicity as well as for the practical virtue of having a smaller, more-easily-explored, one-dimensional parameter space. Note that if these models are to be implemented in the limit of continuous concentrations of particles then it is their corresponding Fokker--Planck equations which must be solved -- this means solving partial differential equations in $d+1$ or $2d+1$ dimensions, respectively.

Both the Brownian and Langevin model can be extended to account for spatial anisotropy, inhomogeneity and the presence of a mean flow, at the cost of increasing the dimension of their parameter spaces; full details are given in~\citet{Berloffii}. Brownian and Langevin dynamics underlie the so-called random displacement and random flight models used for dispersion in the atmospheric boundary layer~\citep{Esler2017}, \changes{and have been applied to the simulation of ocean transport, as models of mixing in the horizontal~\citep{Berloffii}, vertical~\citep{Onink}, and on neutral surfaces~\citep{Reijnders}}.
\citet{Ying} showed how Bayesian parameter inference can be applied to the Brownian model in the inhomogeneous setting using Lagrangian trajectory data. We restrict attention to isotropic
turbulence in this work for simplicity, noting that the methods demonstrated below are equally applicable in the more general case.

\subsection{Data}
For our comparison we consider trajectory data of the form
\begin{align}\label{dataX}
\left\{\left(\bm{X}^{(p)}_0, \cdots, \bm{X}^{(p)}_{N_{\tau}}\right): p\in\left\{1,\cdots, N_p\right\}\right\},
\end{align}
where $\bm{X}^{(p)}_n$ is the position of particle $p$ at time $t=n\tau$. In words, we observe the positions of a set of $N_p$ particles at discrete time intervals of length $\tau$, which we refer to as the sampling time. The performance of the models depends crucially on $\tau$. Since both models are uncorrelated in space, we can rewrite the observations as the set of displacements
\begin{align}\label{dataDX}
\Delta\mathcal{X}_{\tau} &=\left\{\left(\bm{\Delta X}^{(p)}_0, \cdots, \bm{\Delta X}^{(p)}_{N_{\tau}-1}\right): p\in\left\{1,\cdots, N_p\right\}\right\},
\end{align}
where $\Delta\bm{X}^{(p)}_n = \bm{X}^{(p)}_{n+1} - \bm{X}^{(p)}_n$.

In \S\ref{Results} we consider the case that the trajectory data are generated by Langevin dynamics, while in \S\ref{NS2D} we compare the Brownian and Langevin models given data from direct numerical simulations of a forced-dissipative model of stationary, isotropic two-dimensional turbulence.

\section{Methods}\label{Methods}
In this work we appeal to the Bayesian interpretation of probability and statistics. This means that probabilities reflect levels of plausibility in light of all available information. In particular, we  deal with uncertainty in  both the parameters of each model and the models themselves by assigning probabilities to them. We outline this procedure in \S\S\ref{Parameter_inf} and~\ref{Model_inf}.

\subsection{Parameter inference}\label{Parameter_inf}

The goal of parameter inference is to infer the values of the parameters $\bm{\theta}\in\Theta$
of a statistical model, say $\mathcal{M}(\bm{\theta})$, given observational data $\mathcal{D}$.
A model is characterised completely by its likelihood function $p(\cdot|\mathcal{M}(\bm{\theta}))$ which denotes the probability (density) of observations
under $\mathcal{M}(\bm{\theta})$.
Bayesian inference requires the specification of one's belief prior to observations through a prior distribution $p(\bm{\theta}|\mathcal{M})$. One can then invoke Bayes' Theorem,~\eqref{Bayes_thm}, to update this belief in light of the observations. This results in a posterior distribution
\begin{equation}\label{Bayes_thm}
    \overbrace{p(\bm{\theta}|\mathcal{D},\mathcal{M})}^{\text{Posterior}} = \frac{\overbrace{p(\mathcal{D}|\mathcal{M}(\bm{\theta}))}^{\text{Likelihood}}\, \overbrace{p(\bm{\theta}|\mathcal{M})}^{\text{Prior}}}{\underbrace{p(\mathcal{D}|\mathcal{M})}_{\text{Evidence}}},
\end{equation}
which denotes the probability (density) of each $\bm{\theta}\in\Theta$ given observations and prior knowledge~\citep{Jeffreys}. The posterior fully describes the uncertainty in the inferred parameters, in our case $\bm{\theta}= \kappa$ or $\bm{\theta}=(k,\gamma)$. In applications where point estimates of the parameters are required, these can be taken as e.g. the mean or mode of the posterior.

\subsection{Model inference}\label{Model_inf}
Beyond parameter inference we can also make inferences when the model itself, $\mathcal{M}$, is considered unknown. However, in order to meaningfully assign probabilities to models we must assume that the set of models under consideration, $M=\{\mathcal{M}_i\}_{i=1}^{N_M}$, includes all plausibly true models. That is, for any $\mathcal{M}^*\not\in M$, $p(\mathcal{M}^*)=0$. This is known as the $\mathcal{M}$-closed regime (see Chapter 6 of~\cite{BernardoSmith} or~\cite{ClydecIversen}). In situations where all models under consideration are known to be false this assumption appears dubious; however, we note that the same fallacy is committed in Bayesian parameter inference when we assign probabilities to the parameters of a parametric model which we know is imperfect, i.e. false. In the $\mathcal{M}$-closed regime one assigns prior probabilities to models such that $\sum_{i=1}^{N_m}p(\mathcal{M}_i) = 1$. This allows us to again invoke Bayes' Theorem in the form
\begin{align}
    p(\mathcal{M}|\mathcal{D}) = \frac{p(\mathcal{D}|\mathcal{M})\,p(\mathcal{M})}{p(\mathcal{D})}.
\end{align}
If $\mathcal{M}_i$ is parametric with parameters $\bm{\theta}_i\in\Theta_i$, $p(\mathcal{D}|\mathcal{M}_i)$ is given by
\begin{align}\label{evidence}
    p(\mathcal{D}|\mathcal{M}_i) = \int_{\Theta_i} p(\mathcal{D}|\mathcal{M}_i(\bm{\theta}_i))\, p(\bm{\theta}_i|\mathcal{M}_i)\, \text{d}\bm{\theta}_i,
\end{align}
which is known as the model evidence (or marginal likelihood, or model likelihood) of $\mathcal{M}_i$.

An important property of the evidence is that it accounts for parameter uncertainty. Considering the likelihood as a score of model performance given some fixed parameter values, the evidence can be viewed as an expectation of that score with respect to the prior measure on parameters. In this way the evidence favours models where observations are highly probable for the range of parameter values considered plausible a priori. In particular, this means that a model with many parameters which achieves a very high value of the likelihood only for a narrow range of parameter values which could not be predicted a priori is not likely to attain a higher value of the evidence than a model with fewer parameters whose values are better constrained by prior information. This apparent penalty is usually quantified by the so-called \textit{Occam (or Ockham) factor}, named in reference to Occam's razor,
\begin{align}\label{Occam}
    \mathrm{Occam}_i = p(\mathcal{D}|\mathcal{M}_i) / p(\mathcal{D}|\mathcal{M}_i(\bm{\theta}_i^*)) \in [0,1],
\end{align}
where $\bm{\theta}_i^*$ is the posterior mode of $\bm{\theta}_i$~\citep{jaynes_2003, MacKay}.

Given two models, $\{\mathcal{M}_0,\mathcal{M}_1\}$, a test statistic for the hypotheses
\[
\left\{
\begin{array}{ll}
    \mathcal{H}_0:& \mathcal{M}_0\text{ is the true model},\\
    \mathcal{H}_1:& \mathcal{M}_1\text{ is the true model},
\end{array}
    \right.
\]
is given by the Bayes factor~\citep{KassRaftery},
\begin{align}
K_{1,0} = \frac{p(\mathcal{D}|\mathcal{M}_1)}{p(\mathcal{D}|\mathcal{M}_0)},
\end{align}
where a large value of $K_{1,0}$ represents statistical evidence against $\mathcal{H}_0$.

The log-evidence is exactly equal to the log score~\citep{GneitingRaftery}, also known as the ignorance score~\citep{Bernardo1979, BroeckerSmith}, for probabilistic forecasts. Therefore, the log Bayes factor can be understood as a difference of scores for probabilistic models. Merits of the log score have been appreciated since at least the 1950s~\citep{Good}, including its intimate connection with information theory~\citep{RoulstonSmith, Du}. This interpretation of the Bayes factor does not rely on the assumption of the $\mathcal{M}$-closed regime.
In what follows we use the Bayes factor to compare the Brownian and Langevin models.

A useful approximation for the evidence \eqref{evidence} is given by Laplace's method: a Gaussian approximation of the unnormalised
posterior, $p_u(\bm{\theta})=p(\mathcal{D}|\mathcal{M}(\bm{\theta}))\, p(\bm{\theta}|\mathcal{M})$, is obtained from a quadratic expansion of $\ln p_u$ about the posterior mode, $\bm{\theta}^*$,
\begin{align}\label{ln pu}
    \ln\left(p_u(\bm{\theta})\right) \approx \ln\left(p_u(\bm{\theta}^*)\right) - \frac{1}{2}\left(\bm{\theta}-\bm{\theta}^*\right)^{\mathrm{T}} \mathsfbi{J} \left(\bm{\theta}-\bm{\theta}^*\right),
\end{align}
where
\begin{align}\label{info}
    \mathsfbi{J}_{ij} = -\frac{\partial^2}{\partial \bm{\theta}_i\partial \bm{\theta}_j}\ln p_u(\bm{\theta})\bigg\rvert_{\bm{\theta}=\bm{\theta}^*}.
\end{align}
Taking an exponential of~\eqref{ln pu} we  recognise that we have approximated $p_u(\bm{\theta})$ with the probability density function (up to a known normalisation) of a Gaussian random variable with mean $\bm{\theta}^*$ and covariance $\mathsfbi{J}^{-1}$, so \eqref{evidence} becomes
\begin{align}\label{Laplace}
  \underbrace{p(\mathcal{D}|\mathcal{M}_i)}_{\text{Evidence}} \approx \underbrace{p(\mathcal{D}|\mathcal{M}_i(\bm{\theta}_i^*))}_{\text{Maximum likelihood}} \times \hspace{.2cm} \underbrace{p(\bm{\theta}_i^*| \mathcal{M}_i)\left(\text{det}(\mathsfbi{J}/2\pi)\right)^{-\frac{1}{2}}}_{\text{Occam factor}}.
\end{align}
This approximation is accurate for a large number of data points $N_p \times N_\tau$ where a Bernstein--von Mises theorem can be shown to hold, guaranteeing asymptotic normality of the posterior measure~\citep{vanderVaart}.

We highlight that a model's evidence is sensitive to the prior distribution on the parameters, $p(\bm{\theta}|\mathcal{M})$. This is entirely in the spirit of Bayesian statistics in that a parametric model accompanied with the prior uncertainty on its parameters constitutes a single, complete hypothesis for explaining observations. The evidence for a model is less when the mass of prior probability on parameters is less concentrated on those values for which the likelihood is largest.

\section{Results}\label{Results}

In this section we provide details on how BMC can be performed for the Brownian and Langevin model and consider data generated by the Langevin model. We derive the likelihood function for each model, discuss prior distributions for parameters, and the practicalities of inference calculations.

Before we compute the Bayes factor for the Langevin and Brownian models $\mathcal{M}_L$ and $\mathcal{M}_B$, we infer the parameters of both models using a range of datasets with varying sampling time, $\tau$, to establish when each model is \textit{sampling-time consistent} -- we say a model is sampling-time consistent when inferred parameter values are  stable over a range of $\tau$. We emphasise that sampling-time consistency does not imply a model is good, but is certainly a desirable property when one wishes to use a model for extrapolation, e.g. for unobserved values of $\tau$.

Justifications for the Brownian model apply formally only in the large-time limit; we are, therefore, interested in establishing a minimum timescale for the sampling-time consistency of the Brownian model, and further establishing whether the Langevin model, given that it includes time correlation, is sampling-time consistent on shorter timescales. 

Note that in the large-time limit, that is, for $t \gg \gamma^{-1}$, the Langevin dynamics  are asymptotically diffusive: for $\gamma \to \infty$, the
Langevin equations \eqref{Langevin} reduce to \citep{Pav}
\begin{align}\label{diffusive_limit}
    \text{d}\bm{X} = \sqrt{2k}\, \text{d}\bm{W}.
\end{align}
This fact is important when comparing the models, and we  return to it later.

\subsection{Likelihoods}
We can derive explicit expressions for the probability of data of the form of $\Delta \mathcal{X}_{\tau}$ under $\mathcal{M}_B(\kappa)$ and $\mathcal{M}_L(\gamma,k)$ by using their transition probabilities. 
The position increments for $\mathcal{M}_B(\kappa)$ satisfy
\begin{align}\label{Brownian_trans}
    \bm{X}(t+\tau)-\bm{X}(t)\sim \mathcal{N}\left(0,\,2\kappa\tau \mathbb{I}\right),
\end{align}
where $\mathcal{N}(\mu,\, \mathsfbi{C})$ is the $d$-dimensional Gaussian distribution with mean $\mu$ and covariance matrix $\mathsfbi{C}$, and $\mathbb{I}$ is the $d \times d$ identity matrix.
Further, distinct increments are independent under $\mathcal{M}_B(\kappa)$.
Therefore, the desired probability is
\begin{align}
    p\left(\Delta \mathcal{X}_{\tau}|\mathcal{M}_B(\kappa)\right) = \prod_{p=1}^{N_p}\prod_{n=0}^{N_{\tau}-1}\prod_{i=1}^d\rho_{\mathcal{N}}\left(\Delta X^{(p)}_{n,i};\, 0,\, 2\kappa\tau\right),
\end{align}
where $i$ indexes spatial dimension and $\rho_{\mathcal{N}}(\bm{x};\,\bm{\mu},\mathsfbi{C})$ is the probability density at $\bm{x}$ of
the Gaussian distribution $\mathcal{N}(\bm{\mu},\, \mathsfbi{C})$.

The corresponding likelihood for the Langevin model is shown in appendix~\ref{Lang_MargLike} to be
\begin{align}\label{Lang_Marg}
    p(\Delta \mathcal{X}_{\tau}|\mathcal{M}_L(\gamma, k)) &= \prod_{p=1}^{N_p}\prod_{i=1}^d \rho_{\mathcal{N}}\left((\Delta X_{0,i}^{(p)}, \cdots \Delta X_{N_{\tau}-1,i}^{(p)})^{\mathrm{T}};\,\bm{0} ,\, \mathsfbi{S}\right),
\end{align}
where $\mathsfbi{S}$ is the symmetric Toeplitz matrix with
\begin{align}
    \mathsfi{S}_{ij} = \left\{\begin{array}{cc}
        2k\tau(1-\varphi(\gamma \tau)) &  \text{if $m=0$}\\
        k\gamma\tau^2\varphi^2(\gamma\tau)\e^{-(m-1)\gamma\tau} & \text{if $m>0$}
    \end{array}\right. ,
\end{align}
\begin{align}
    \varphi(x) = \frac{1-\e^{-x}}{x},
\end{align}
and $m=|i-j|$.

\subsection{Prior distributions}\label{priors}
It is necessary, both for parameter and model inference, to specify a prior distribution for each of the parameters, $\kappa$, $\gamma,$ and $k$. For a given flow we can appeal to scaling considerations to assign a prior mean to each parameter, derived from characteristic scales. Once such prior means are prescribed, the maximum entropy principle, along with positivity and independence of the parameters
motivates a choice of corresponding exponential distributions as priors \citep{jaynes_2003,CoverThomas}. That is, for a parameter $\theta>0$ with prior mean $\mu$, the distribution with maximum  entropy
is the exponential distribution $\mathrm{Exp}(\lambda)$ with rate $\lambda=1/\mu$. 
We use this prescription for our choice of prior.

\subsection{Inference numerics}
The computations we perform for Bayesian parameter inference are: (i) an optimisation procedure to find the posterior mode, $\bm{\theta}^*$, and (ii) a single evaluation of the Hessian of the log-posterior distribution at $\bm{\theta}^*$, $-\mathsfbi{J}$ in \eqref{info}, which we can use to estimate the posterior variance by a Gaussian approximation as in~\eqref{ln pu}. We have analytical expressions for the likelihood and prior for both models, so we can easily evaluate the negative log unnormalised posterior, $f(\bm{\theta}) = -\ln p_u(\bm{\theta}|\mathcal{D})$, in each case; we find $\bm{\theta}^*$ by minimising $f(\bm{\theta})$ using the \texttt{SciPy} function \texttt{optimize.minimize()}.

In the case of the Brownian model derivatives of $f(\bm{\theta})$ are easily derived analytically, so we use the \texttt{L-BFGS-B} routine which exploits gradient information and allows for the specification of lower bound constraints to enforce positivity~\citep{L-BFGS-B}.
In the case of the Langevin model calculation of derivatives of the posterior is nontrivial because the likelihood~\eqref{Lang_Marg} is a complicated function. For this reason we use the gradient-free \texttt{Nelder-Mead}~\citep{Nelder-Mead} routine rather than \texttt{L-BFGS-B}. We evaluate $\mathsfbi{J}^{-1}$ approximately using a fourth-order central difference approximation for the log-likelihood.

No further computations are required for BMC if the Laplace's method approximation for the evidence in~\eqref{Laplace} is used.

\begin{figure}
    \centering
    \includegraphics[width=4in]{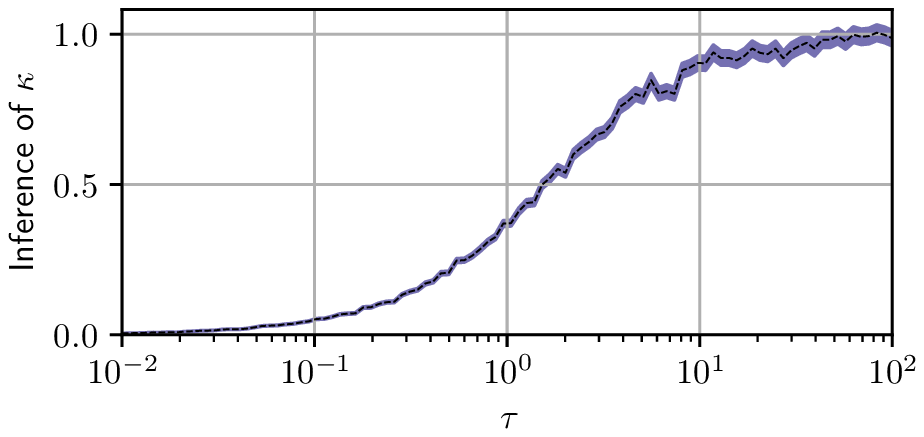}
    \includegraphics[width=4in]{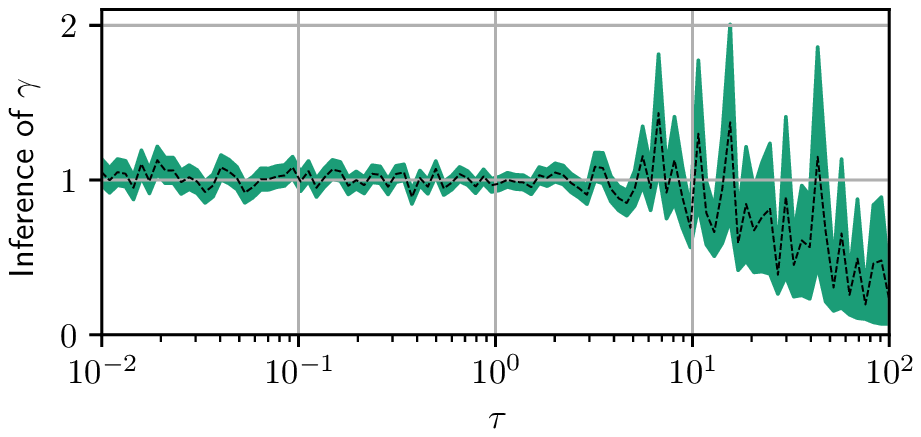}
    \includegraphics[width=4in]{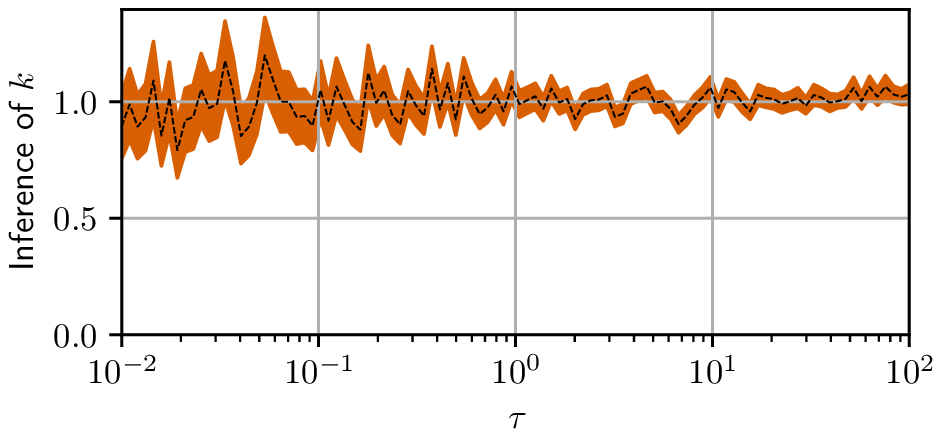}
    \caption{Parameter inference for the Brownian and Langevin models as a function of observation interval, $\tau$, for data from the Langevin model in three spatial dimensions. Dashed lines indicate posterior mode estimates, $\theta^*=\kappa$ (top), $\gamma^*$ (middle) and $k^*$ (bottom); shaded areas show $\theta^*\pm \text{SD}(\theta|\Delta\mathcal{X}_{\tau})$. Each inference is made with a fixed volume of data: $N_p=100$ and $N_{\tau}=10$.}
    \label{fig:params_Langevin}
\end{figure}

\begin{figure}
    \centering
    \includegraphics[width=4in]{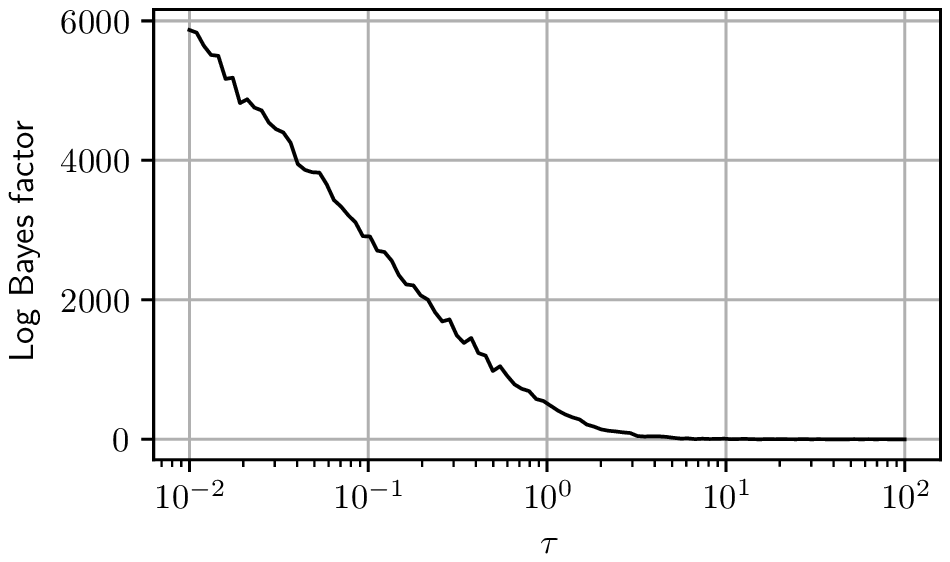}
    \includegraphics[width=4in]{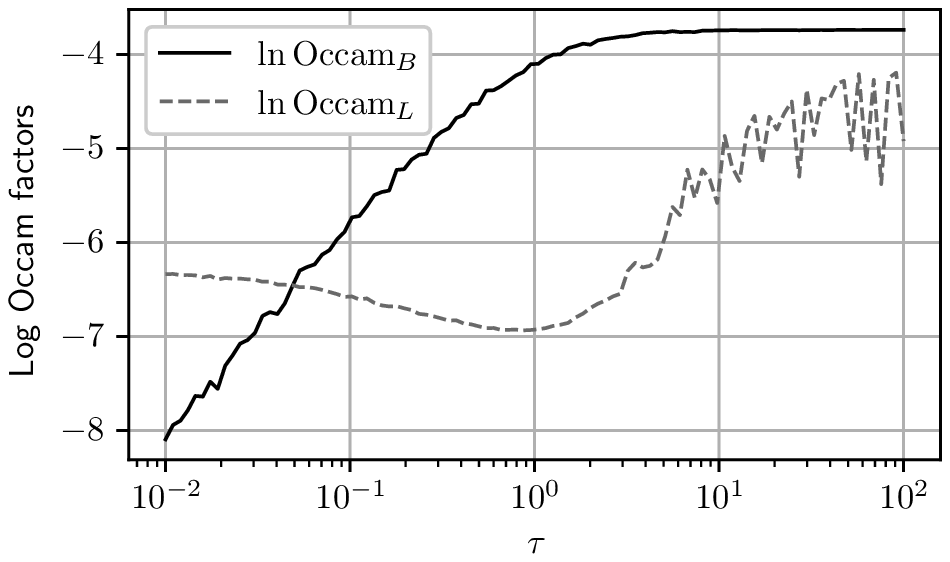}
    \caption{Log Bayes factors, $\ln K_{L,B}$, and corresponding log Occam factors, as a function of $\tau$, given the same data used for figure~\ref{fig:params_Langevin}.}
    \label{fig:BMC_Langevin}
\end{figure}

\subsection{Test case: Langevin data}\label{Testcase}
As a test case and to build intuition, we first consider  trajectory data generated by the Langevin model with $d=3$. In this case, one of the two candidate models  is the true model.
We  generate the data  by simulating the Langevin SDE~\eqref{Langevin} exactly, drawing initial velocities from the stationary distribution $\bm{U}|\mathcal{M}_L(\gamma, k)\sim\mathcal{N}(\bm{0},\gamma k\mathbb{I})$, and using the transition probabilities~\eqref{Langevin_trans}; velocity data are discarded to construct the dataset of position increments $\Delta\mathcal{X}_{\tau}$.

We set $\gamma = k = 1$, fix $N_p=100$ and $N_{\tau}=10$, and perform Bayesian parameter inference and model comparison with a series of independently generated datasets with $\tau\in[10^{-2},10^2]$. We set fixed priors $\gamma,k,\kappa\sim\text{Exp}(1)$.

Figure~\ref{fig:params_Langevin} shows the results of the parameter inference. Note that both Langevin parameters are well identified until, at sufficiently large $\tau$, the error of the posterior mode estimate of $\gamma$ grows along with the posterior standard deviation of $\gamma$. This is a manifestation of the diffusive limit \eqref{diffusive_limit} of the Langevin dynamics, wherein $\bm{X}(t+\Delta t)-\bm{X}(t)\sim\mathcal{N}(\bm{0}, 2k\Delta t\mathbb{I})$ is independent of $\gamma$. Unsurprisingly, then, $\Delta\mathcal{X}_{\tau}$ is less informative about $\gamma$ when $\tau$ is large.

The diffusivity $\kappa$ of the Brownian model is sampling-time consistent only when $\tau$ is sufficiently large, i.e. in the diffusive limit of the Langevin dynamics, when $\kappa\approx k$. The inaccuracy of posterior mode estimates of $\kappa$ at small $\tau$ is expected as it is known that the inference of diffusivities from discrete trajectory data is sensitive to sampling time~\citep{Cotter_Pavliotis2009}. We note that $\gamma^{-1}=1$ is the decorrelation time for this data so that the timescales at which this limiting behaviour is observed, $\tau\gtrsim 10$, are indeed large.

Note that the posterior mode estimates of $\gamma$ eventually decay to zero as $\tau$ increases; since, as observed, $\Delta\mathcal{X}_{\tau}$ becomes less informative about $\gamma$ with increasing $\tau$, the contribution of the prior information to the posterior becomes dominant over the contribution from the likelihood -- the consequence of this is that the posterior mode tends to the prior mode, which is zero since we take $\gamma\sim\text{Exp}(1)$.

Figure~\ref{fig:BMC_Langevin} shows the log Bayes factors found using  Laplace's method  for the evidences. We see that for a significant range of $\tau$ the Langevin model is preferred, indicated by large positive values of $\ln K_{L,B}$, but its dominance diminishes as $\tau$ increases until the diffusive limit is reached, at which point values of $|\ln K_{L,B}|< 1$ are typical, indicating insubstantial preference for either model.

Also shown in figure~\ref{fig:BMC_Langevin} are the corresponding log Occam factors. Occam factors for the Brownian model are approximately constant once $\tau$ is sufficiently large, while the Occam factors for the Langevin model increase at large $\tau$ in line with a broadening posterior. This is indicative of decreased sensitivity to choice of parameters, specifically $\gamma$, whose value becomes less critical for explaining dynamics on large timescales.

\section{Application to two-dimensional turbulence}\label{NS2D}
 In this section we report an application of BMC to particle trajectories in a 
 model of stationary, isotropic two-dimensional (2D) turbulence.

\subsection{Forced-dissipative model}
\label{F-D}

We consider a forced--dissipative model of isotropic 2D turbulence in an incompressible fluid governed by the vorticity equation~\citep{Vallis_2017}
\begin{align}\label{baro}
    \frac{\partial \zeta}{\partial t} + (\bm{u}\cdot\nabla)\zeta = F + D,
\end{align}
where $\zeta$ is the vertical vorticity and $F$ and $D$ represent forcing and dissipation, respectively. The particular forcing used is an additive homogeneous and isotropic white Gaussian noise concentrated in a specified range of  wavenumber centred about a forcing wavenumber, $k_F$. In particular, following~\citet{Scott}, we have that, at each timestep, the Fourier transform  of $F$ satisfies
\begin{align}
    \text{Re}\left(\hat{F}(\bm{k})\right) \stackrel{d}{=} \text{Im}\left(\hat{F}(\bm{k})\right)\sim \mathcal{N}
    \left(0, \frac{A_{F}\mathcal{F}_F(|\bm{k}|)}{2\pi|\bm{k}|}\right),
\end{align}
where $A_F$ is the forcing amplitude, and  $\mathcal{F}_F(|\bm{k}|) = 1$ for $\left|\left|\bm{k}\right|-k_F\right|\leq2$ and $\mathcal{F}_F= 0$ otherwise.

Two dissipation mechanisms are included: (i) small-scale dissipation implemented with a scale-selective exponential cut-off filter (see~\cite{ArbicFlierl} for details and justification), and (ii) large-scale friction (aka hypodiffusion), so that total dissipation is given by
\begin{align}
    D = A_{\mathrm{lsf}}\psi + \mathrm{ssd},
\end{align}
where ssd denotes the small-scale dissipation.

Equation~\eqref{baro} is solved in a periodic domain, $[0,2\pi]^2$, using a standard pseudospectral solver, at a resolution of $1024\times1024$ gridpoints, with the third-order Adams--Bashforth timestepping scheme. The complete set of flow configuration parameter values for our simulations are given in Table~\ref{tab:flow_params}.

\begin{table}
  \begin{center}
    domain $= [0,2\pi]^2$; \quad resolution = 1024 $\times$ 1024 grid points; \quad  timestep size $=2.5\times10^{-4}$;\\
    $k_F=64$; \quad $A_F=8.9\times10^8$; \quad $A_{\text{lsf}}=1$.
  \caption{Flow configuration parameter values for simulations of the 2D turbulence model.}
  \label{tab:flow_params}
  \end{center}
\end{table}

The model is initialised with a random, Gaussian field with prescribed mean energy spectrum and is run until the total energy,
\begin{align}
    E(t) := \frac{1}{2} \int |\bm{u}(\bm{x},t)|^2\, \text{d}x\,\text{d}y,
\end{align}
appears to reach a statistically stationary state; this amounted to a spin-up time of approximately $6800$ eddy turnover times, where the eddy turnover time is estimated by 
\begin{equation}
    \tau_{\zeta} = 2\pi/\sqrt{Z},
\end{equation}
and $Z$ is the total enstrophy,
\begin{align}
    Z := \frac{1}{2} \int \zeta^2\, \text{d}x\,\text{d}y.
\end{align}

Figure~\ref{fig:z} shows a snapshot of the vorticity field at the end of the spin-up process. Enstrophy is concentrated in a population of coherent vortices whose scale is set by the forcing scale, $k_F^{-1}$. Figure~\ref{fig:ke_spec} shows the isotropic energy spectrum calculated at the same instant. A power law of approximately $k^{-2}$ is observed at wavenumbers between the peak wavenumber at $k\approx 6$ and the forcing wavenumber at $k\approx 64$ (indicated with a vertical line). A second power law of approximately $k^{-3.5}$ is seen at wavenumbers between the forcing scale and the dissipation range. Large-scale friction prevents the indefinite accumulation of energy at the largest scales, while continued forcing prevents energy from concentrating exclusively around a peak wavenumber at late times, and, by inputting enstrophy at a moderate scale, sustains a lively population of vortices.



\begin{figure}
    \centering
    \includegraphics[width=5.0in]{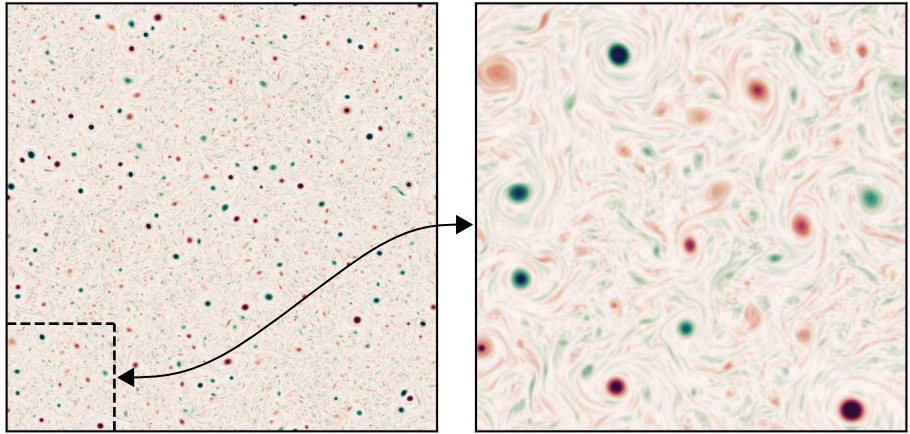}
    \caption{Snapshot of the vorticity field in the forced-dissipative model at stationarity showing $x,y\in[0,2\pi]$ (left) and $x,y\in[0,\pi/2]$ (right).}
    \label{fig:z}
\end{figure}

\begin{figure}
    \centering
    \includegraphics[width=3in]{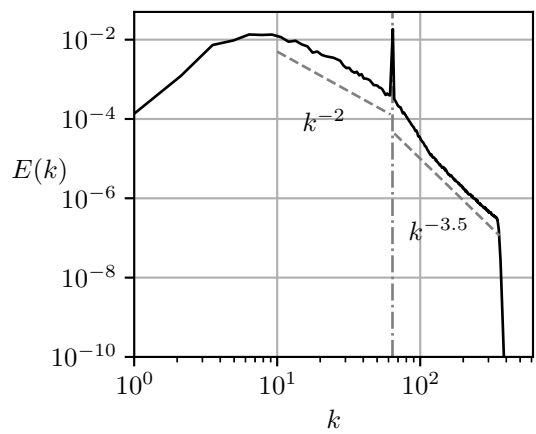}
    \caption{Snapshot of the isotropic energy spectrum in the forced--dissipative model at stationarity. }
    \label{fig:ke_spec}
\end{figure}

\subsection{Particle numerics}
After spin-up, a set of $1000$ passive tracer particles are evolved in the flow of the forced-dissipative model for approximately another $6800$ eddy turnover times; this is done using bilinear interpolation of the velocity field and the fourth-order Runge--Kutta time-stepping scheme. Particles are seeded at initial positions chosen uniformly at random in the domain.

Figure~\ref{fig:traj} shows a subset of the trajectory data generated. Some particles follow highly oscillatory paths, while others do not, depending on whether they are seeded in the interior of a coherent vortex or in the background turbulence.

\begin{figure}
    \centering
    \includegraphics[height=3in]{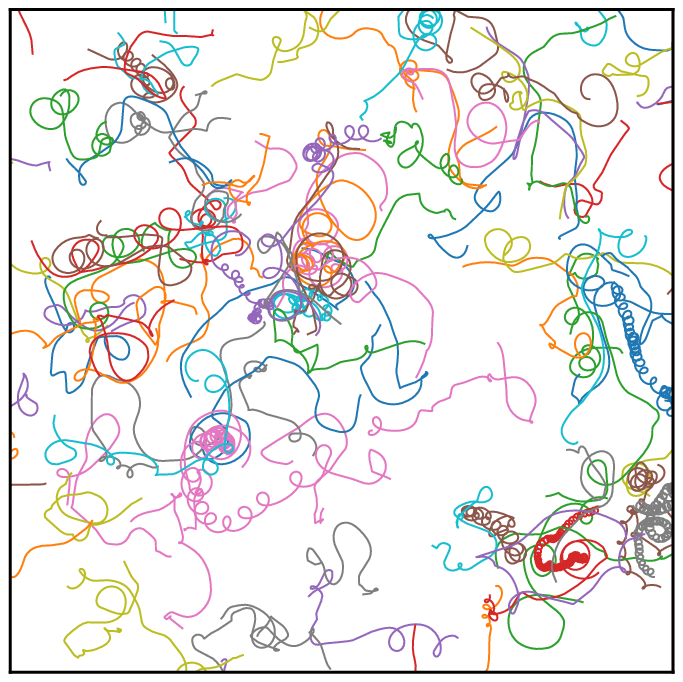}
    \caption{Trajectories of $100$ passive particles advected in the the forced--dissipative model, shown as recorded over a period of $100\tau_{\zeta}$ with a different colour for each trajectory.}
    \label{fig:traj}
\end{figure}

\subsection{Diagnostics}
\changes{To illustrate the dynamics that we parameterise with the stochastic models we show two diagnostics commonly used in Lagrangian analyses~\citep{Pasquero, vanSebille}, namely, the Lagrangian velocity autocovariance function (LVAF), defined in the isotropic case as}
\begin{align}
    r(\tau) = \langle U^{(p)}(t+\tau) U^{(p)}(t) \rangle,
\end{align}
\changes{where the angled-brackets denote the average over $t$ and particles $p$,
and the absolute diffusivity}
\begin{align}
    \kappa_{\mathrm{abs}}(\tau) = \frac{\left\langle\left(\Delta X^{(p)}(\tau)\right)^2\right\rangle}{2\tau}.
\end{align}

\changes{Figure~\ref{fig:r} shows the LVAF as estimated from the simulated particle trajectory data. The corresponding LVAF of the Brownian model is a delta function at zero, since velocity is implicitly represented as a white-noise process, while the LVAF of the Langevin model, which represents particle velocity as an Ornstein--Uhlenbeck process, is}
\begin{align}
    r_{\text{OU}}(\tau; \gamma, k) = k\gamma\exp(-\gamma|\tau|).
\end{align}
\changes{In contrast, the estimated LVAF of the forced--dissipative model not only shows finite decorrelation time but is noticeably sub-exponential.}

\changes{Figure~\ref{fig:diff} shows the estimated absolute diffusivity. In line with the asymptotic laws described in~\citet{Taylor} the absolute diffusivity is linear at small $\tau$ and constant at large $\tau$, corresponding to the ballistic and diffusive regimes, respectively. The absolute diffusivity of the Brownian model is constant, while that of the Langevin model is}
\begin{align}
    \kappa_{\text{OU}}(\tau) = k\left(1-\varphi(\gamma\tau)\right),
\end{align}
\changes{which is linear at small $\tau$ and constant at large $\tau$.}

\changes{Qualitatively, from comparing these diagnostics with those of the stochastic models it is clear that on sufficiently large times (in the diffusive regime) the Brownian model is valid; in particular, the LVAF is well-approximated by a delta function at large-times, and correspondingly the absolute diffusivity is constant. On timescales shorter than the diffusive regime the LVAF of the observed trajectories may be better approximated by that of the Langevin model; however, the quality of this approximation is in general unclear a priori. It could be tempting to estimate $\gamma$ by fitting the LVAF, using e.g. a least-squares method, but this approach would not correctly deal with uncertainty in parameters.
}


\begin{figure}
    \centering
    \includegraphics{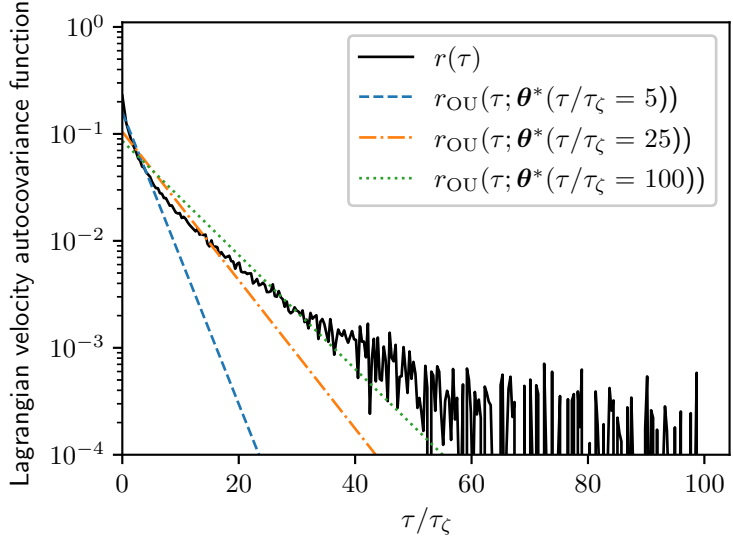}
    \caption{\changes{LVAF $r(\tau)$ for the forced--dissipative model, as estimated from the full set of $1000$ simulated particle trajectories. The LVAF of the Langevin model $r_{\mathrm{OU}}(\tau)$ is also shown using MAP estimates (discussed below) $\bm{\theta}^*=(\gamma^*, k^*)$ derived from datasets with $\tau = (5, 25, 100) \tau_{\zeta}$, respectively (see figure~\ref{fig:params_NS2D}).}}
    \label{fig:r}
\end{figure}
\begin{figure}
    \centering
    \includegraphics{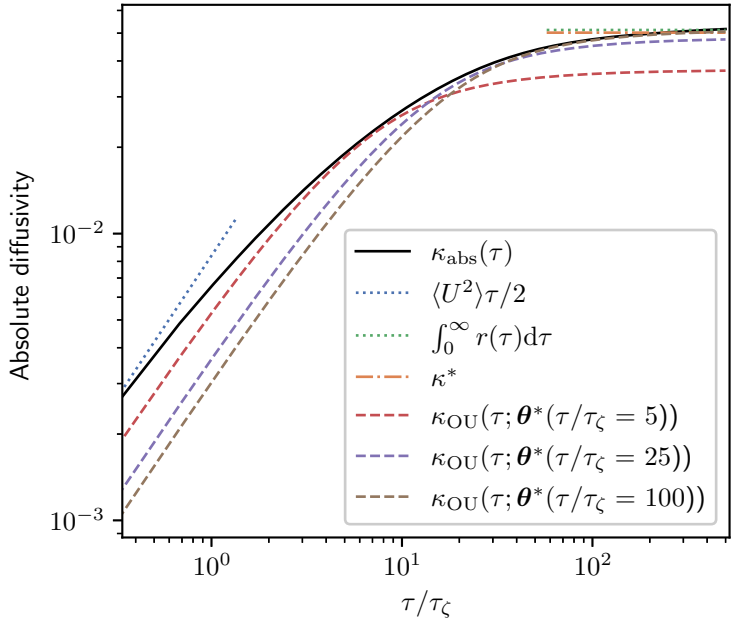}
    \caption{\changes{Absolute diffusivity, $\kappa_{\mathrm{abs}}(\tau)$, for the forced--dissipative model, as estimated from the full set of $1000$ simulated particle trajectories. A MAP estimate $\kappa^*$ is shown, along with two asymptotic laws: $\kappa_{\mathrm{abs}}(\tau)= \mathrm{linear}$ (ballistic regime), and $\kappa_{\mathrm{abs}}(\tau)= \mathrm{const.}$ (diffusive regime).}}
    \label{fig:diff}
\end{figure}

\subsection{Parameter inference and BMC}\label{NS2D_inference}
We now apply the parameter inference and BMC procedures demonstrated in the test case of \S\ref{Testcase}. By subsampling the results of our particle simulations we generate datasets with $N_p=1000$, $N_{\tau}=25$, and a set of sampling times $\tau$ in the range $[\tau_{\zeta}, 250\tau_{\zeta}]$.

Prior means for the parameters are derived from $\tau_{\zeta}$ and the root-mean-square velocity $u_{\text{RMS}}=\sqrt{2E}$, where $E$ is the mean energy: as discussed in \S\ref{priors} we take
\begin{subequations}\label{means}
\begin{align}
    \mathbb{E}[\kappa] = \mathbb{E}[k] &= u_{\text{RMS}}^2 \tau_{\zeta},\\
   \mathbb{E}[\gamma] &= \tau_{\zeta}^{-1},
\end{align}
\end{subequations}
and use the corresponding exponential distributions as priors.

The results of the parameter inference are shown in figure~\ref{fig:params_NS2D}. The Brownian model is sampling-time consistent for $\tau \gtrsim 150 \tau_{\zeta}$, with a posterior mode that differs by  \changes{$40\%$} from the scaling estimate used as prior mean. The long time required for sampling-time consistency is in line with the expected validity of the Brownian model in the long-time limit. \changes{In this limit $\kappa^*$ agrees very well with \citeauthor{Taylor}'s \citeyear{Taylor} theoretical prediction, $\kappa=\int_0^\infty r(\tau)\text{d}\tau$ (see figure~\ref{fig:diff}).}

\changes{The Langevin model is \changes{roughly} sampling-time consistent from much smaller values of $\tau$, say $\tau \gtrsim 50\tau_{\zeta}$. This suggests that there is a range of sampling times, roughly {$50 \tau_{\zeta} \lesssim \tau \lesssim 150 \tau_{\zeta}$}, where the Langevin model is potentially useful but the Brownian model is not. The BMC analysis below sheds further light on this. However, there is noticeable decay in the MAP estimates of $\gamma$ with increasing $\tau$ --  this is likely a reflection of the sub-exponential nature of the true LVAF. In figure~\ref{fig:r} we plot the Langevin LVAF given parameters inferred with data of various $\tau$, where the decay in estimates of $\gamma$ corresponds to a shallowing of the Langevin LVAF. In figure~\ref{fig:diff} we plot the absolute diffusivity of the Langiven model $\kappa_{\text{OU}}$ using the same parameter estimates as in figure~\ref{fig:r}. The absolute diffusivity is best approximated at a timescale matching the sampling time of the data.} The posterior mode of $\gamma$, when roughly stable, is almost an order of magnitude smaller than the scaling estimate in  \eqref{means}, indicating that particle dynamics decorrelate slower than might be predicted by a naive dimensional analysis based on the enstrophy alone. In particular, the inferred value of $\gamma$ corresponds to a decorrelation time of about $8$ eddy turnover times. As in the test case of \S\ref{Testcase} the posterior standard deviation of $\gamma$ grows with $\tau$ as the diffusive limit is reached and the particle dynamics become insensitive to $\gamma$. It is interesting to note that the Langevin diffusivity $k$ is estimated consistently for sampling times much shorter than those required to estimate the Brownian diffusivity $\kappa$ even though their values are identical when the Brownian model represents the dispersion well. This suggests that carrying out Bayesian inference of the Langevin model might provide a means to estimate the Brownian diffusivity when data is not available over the long, diffusive time scales that are required a priori. \changes{This may generalise to other flows only when the Langevin model is a reasonable approximation -- inference of $k$ is unlikely to be sampling time consistent if inference of $\gamma$ is not, for example, due to the LVAF of the flow of interest being very far from exponential.} We emphasise that the inference results just described are largely insensitive to specification of the prior.

\begin{figure}
    \centering
    \includegraphics[width=4in]{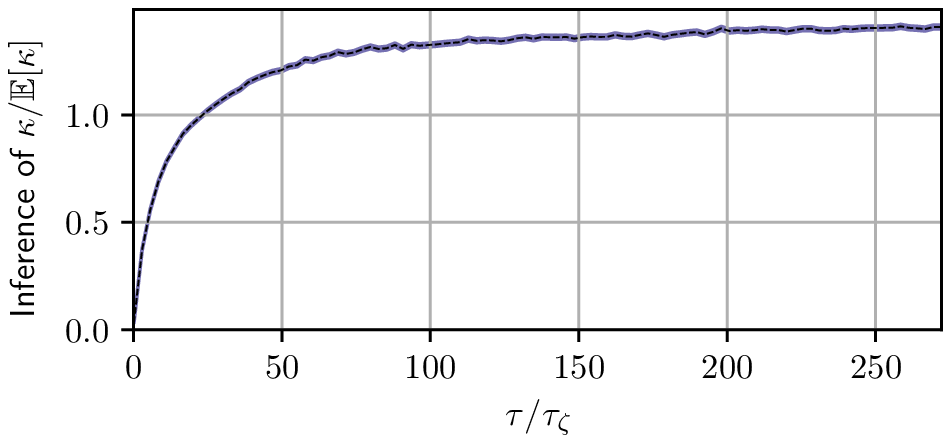}
    \includegraphics[width=4in]{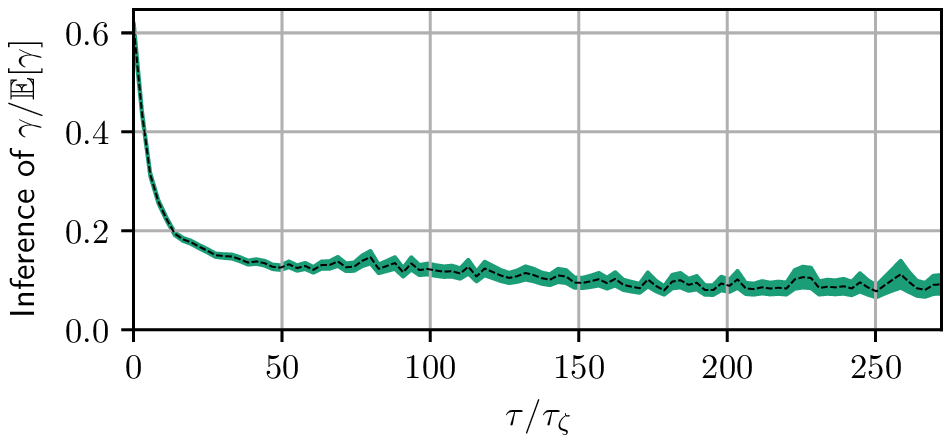}
    \includegraphics[width=4in]{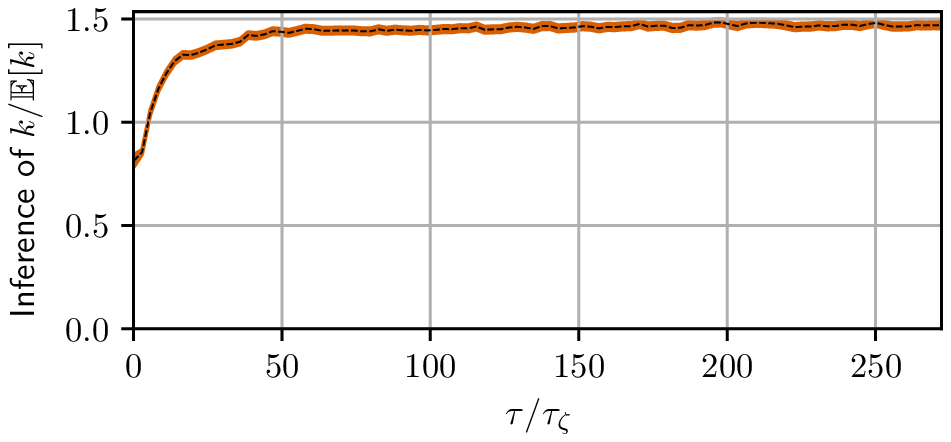}
    \caption{Parameter inference for the Brownian and Langevin models as a function of observation interval, $\tau$, for data from the two-dimensional turbulence model. Dashed lines indicate posterior mode estimates, $\theta^*$, normalised with respect to prior means, and shaded areas are $\theta^*\pm \text{SD}(\theta|\Delta\mathcal{X}_{\tau})$. Each inference is made with a fixed volume of data: $N_p=1000$ and $N_{\tau}=25$.}
    \label{fig:params_NS2D}
\end{figure}

\begin{figure}
    \centering
    \includegraphics[width=4in]{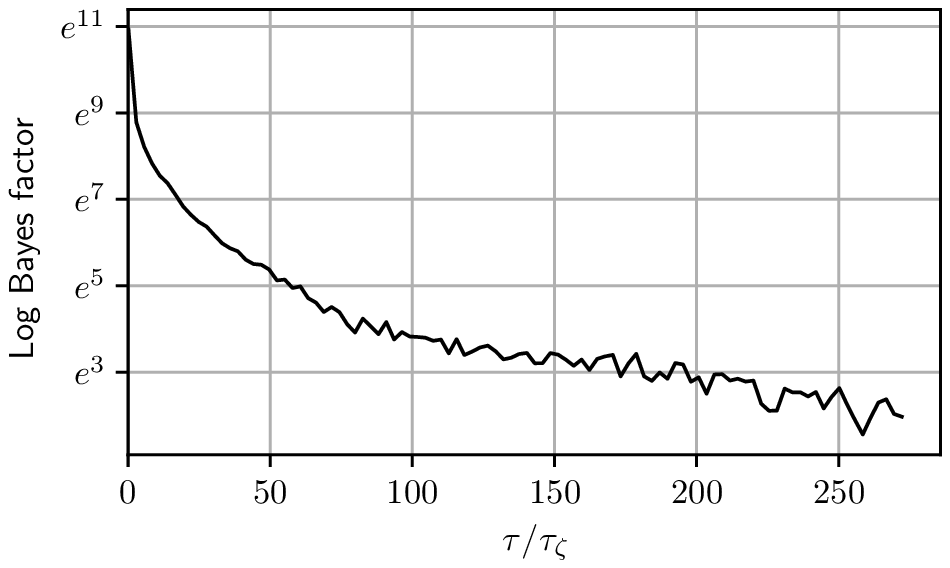}
    \includegraphics[width=4in]{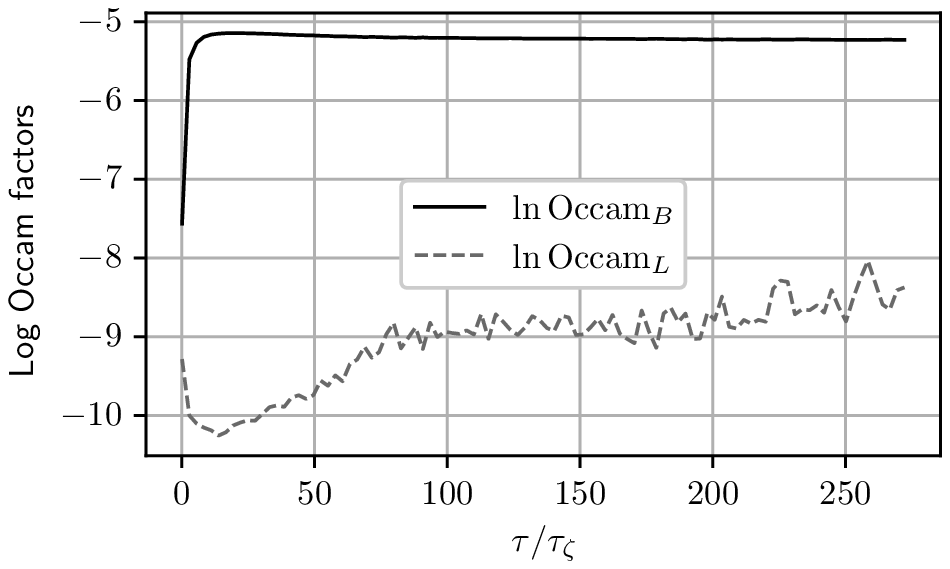}
    \caption{Log Bayes factors, $\ln K_{L,B}$, and corresponding log Occam factors, as a function of $\tau$, given the same data used for figure~\ref{fig:params_NS2D}.}
    \label{fig:BMC_NS2D}
\end{figure}

The results of the BMC for the turbulence model data are shown in figure~\ref{fig:BMC_NS2D}. The picture is similar to that in the test case of \S\ref{Testcase}, in that the Bayes factor favours the Langevin model for shorter timescales, but with diminishing strength as $\tau$ is increased, until, at timescales corresponding to convergence of the Brownian diffusivity, the value of the log Bayes factor becomes small enough that the two models cannot be meaningfully discriminated.

Also shown in figure~\ref{fig:BMC_NS2D} are the corresponding log Occam factors. For $\tau$ large enough that the Brownian model is sampling-time consistent, its Occam factor is approximately constant and larger than that of the Langevin model. As in the test case in \S\ref{Testcase}, the Occam factor for the Langevin model increases towards that of the Brownian model at large $\tau$ when the particle dynamics are sufficiently decorrelated that the likelihood is less sensitive to the value of $\gamma$, \changes{albeit more slowly, owing to the more slowly decaying LVAF of the turbulent dynamics}. The difference in log Occam factors is much smaller than the difference in the corresponding maximum log-likelihoods for all but the largest values of $\tau$, which explains why the Bayes factor mainly favours the Langevin model. 

In summary, these results indicate that while the Brownian model is adequate on sufficiently large timescales($\tau\gtrsim150\tau_{\zeta}$), the Langevin model can explain better the dynamics of tracer particles in the turbulence model of \S\ref{F-D} on shorter timescales ($50\tau_{\zeta}\lesssim\tau\lesssim150\tau_{\zeta}$). On time scales $\tau\gtrsim150\tau_{\zeta}$ the two models are indistinguishable in their performance, so that in this regime the Brownian model should be favoured in practice as a more parsimonious description.

\section{Conclusions}\label{Conclusions}
We have demonstrated the application of BMC to a problem of interest in fluid dynamics, and shown that we can compare the performance of competing stochastic models of particle dynamics given discrete trajectory data alone while accounting for parameter uncertainty. In particular, we found that the Langevin model is preferred over the Brownian model for describing particle dynamics in a model of two-dimensional turbulence on a range of timescales, but that on sufficiently large timescales the two models perform equally well.

The broad conclusion of the BMC, then, is that the additional complexity of the Langevin model, associated with the presence of an additional parameter, is justified: its better capability to explain the data, as quantified by the maximum likelihood, overwhelms the penalty for complexity quantified by the Occam factor. We stress, however, that this conclusion does not take into account the computational cost involved if the models are used for predictions.

The application of the BMC method to other problems is limited by the feasibility of the calculation of the model evidence. Specifically, BMC inherits the usual challenges of Bayesian and likelihood-based inference procedures, namely that the likelihood can be intractable or expensive to compute for complex models -- the Brownian and Langevin models considered here, as linear stochastic differential equations, are very simple examples whose likelihoods could be computed analytically -- alternative models which are nonlinear, have higher dimension, or have more complicated correlation structure will likely have intractable likelihoods.
For example, for spatially inhomogeneous flows, such as in the atmosphere or oceans, nonlinear models arise with spatially-varying (and hence high-dimensional) parameters.
Fortunately, the collection of methods referred to as approximate Bayesian computation have been developed to deal with this problem. For example,~\cite{Carson} used the SMC$^2$ (`sequential Monte Carlo squared') algorithm to compare SDE models of glacial--interglacial cycles with intractable likelihoods.

There is the further issue of performing the integration required to obtain the evidence as in~\eqref{evidence}. When the posterior is sufficiently Gaussian-like, i.e. peaked around a single mode, Laplace's method can be very accurate~\citep{KassRaftery} as well as cheap, however, this requires (at least an approximation to) the Hessian of the log-posterior at its mode. \changes{Aside from Laplace's method, \citet{Krog2017} and \citet{Thapa2018} have implemented the nested sampling algorithm of Skilling~\citep{Skilling} to calculate the evidence in similar analyses}, while the line of work by~\cite{hannart2016dada},~\cite{carrassi2017estimating} and~\cite{metref2019estimating} has sought to perform model evidence estimation using ensemble-based data assimilation methods originally designed for state estimation in the context of incomplete, noisy observations of high-dimensional dynamical systems.

While BMC inevitably comes with computational challenges in complex problems, there are many cases where it can feasibly be applied, and, where it cannot, it should serve as a useful theoretical starting point, with alternative methods measured by how well their conclusions agree with those of BMC.


\small{
\vspace{0.5em} \noindent \textbf{Funding.} M.B. was supported by the MAC-MIGS Centre for Doctoral Training under EPSRC grant EP/S023291/1.
J.V. was supported by the EPSRC Programme Grant EP/R045046/1: Probing Multiscale Complex Multiphase Flows with Positrons for Engineering and Biomedical Applications (PI: Prof. M. Barigou, University of Birmingham).

\vspace{0.5em} \noindent \textbf{Declaration of interests.} The authors report no conflict of interest.

\vspace{0.5em} \noindent \textbf{Data availability statement.} The code required to reproduce the results herein is available at \href{https://doi.org/10.5281/zenodo.5820320}{doi.org/10.5281/zenodo.5820320} and the trajectory data generated with the 2D turbulence model is available at \href{https://doi.org/10.7488/ds/3267}{doi.org/10.7488/ds/3267}.
}

\appendix

\section{Langevin likelihood for position observations}\label{Lang_MargLike}

To derive $p(\Delta \mathcal{X}_{\tau}|\mathcal{M}_L(\gamma, k))$ we simplify notation by recognising that all particles are independent under $\mathcal{M}_L$ and that dynamics in each spatial dimension are independent. We therefore need only calculate $p(\Delta \mathcal{X}_{\tau}|\mathcal{M}_L(\gamma, k))$ in the one-dimensional, single-particle case. We proceed by: (i) showing that the joint process of particle position and velocity is an order-one vector autoregressive process, or VAR(1) process, and hence, has a Gaussian likelihood, (ii) calculating the mean and covariance for a sequence of joint position--velocity observations, and (iii) marginalising this likelihood to find $p(\Delta \mathcal{X}_{\tau}|\mathcal{M}_L(\gamma, k))$.

It can be shown that for the one-dimensional Langevin equation
\begin{align}\label{Langevin_trans}
    \bm{Y}_n|\bm{Y}_{n-1}
    \sim \mathcal{N}\left(
    \begin{pmatrix}
    U_n\varphi(\gamma\tau)\tau \\
    U_n \e^{-\gamma\tau}
    \end{pmatrix},\,
    \mathsfbi{C}
    \right),
\end{align}
where $\bm{Y}_n:=(\Delta X_n,U_{n+1})^{\mathrm{T}}$ and 
\begin{subequations}
\begin{align}\label{C_ij}
    \mathsfi{C}_{11}&= 2k\tau\left(1 - 2\varphi(\gamma\tau)+\varphi(2\gamma\tau)\right),\\
    \mathsfi{C}_{12}&=\mathsfi{C}_{21}=k(\varphi(\gamma\tau)\gamma\tau)^2,\\
    \mathsfi{C}_{22} &= 2k\gamma^2\tau\varphi(2\gamma\tau).
\end{align}
\end{subequations}
This follows from the well-known solution of the Ornstein--Uhlenbeck process,
\begin{align}
    U(t) = U(0)\e^{-\gamma t} + \gamma\sqrt{2k}\int_0^t \e^{-\gamma(t-t')}\,\text{d}W(t')
\end{align}
and the corresponding solution for the position,
\begin{subequations}
\begin{align}
    X(t) &= X(0) + \int_0^t U(t')\,\text{d}t'\\
    &= X(0) + U(0)\varphi(\gamma t)t-\sqrt{2k}\int_0^t\e^{-\gamma(t-t')}\,\text{d}W(t')+\sqrt{2k}\,W(t).
\end{align}
\end{subequations}

Therefore, we can write the Langevin model in the time-discretised form
\begin{align}\label{Langevin_VAR}
    \bm{Y}_n
    =
    \mathsfbi{A}
    \bm{Y}_{n-1}
    + \bm{\varepsilon}_n,
\end{align}
 where 
 \begin{align}
      \mathsfbi{A} =
    \begin{pmatrix}
    0 & \varphi(\gamma\tau)\tau \\
    0 & \e^{-\gamma\tau}
    \end{pmatrix},
 \end{align}
 and  $\bm{\varepsilon}_n$ is a mean-zero white-noise process with covariance matrix $\mathsfbi{C}$.

The discrete process~\eqref{Langevin_VAR} has the form of a VAR(1) process.
Furthermore, $\mathsfbi{Y}_n$ is stationary with mean and stationary variance
\begin{align}\label{G0}
    \bm{\mu} = 
    \begin{pmatrix}
    0\\
    0
    \end{pmatrix},\,\,
    \mathsfbi{V}=
    \begin{pmatrix}
    2k\tau(1-\varphi(\gamma \tau)) & k\varphi(\gamma\tau)\gamma\tau\\
    k\varphi(\gamma\tau)\gamma\tau & k\gamma
    \end{pmatrix}.
\end{align}
To see this, note that the marginal distribution of $U(t)$ at any time is given by the stationary distribution of the Ornstein--Uhlenbeck process,
\begin{align}\label{U_stat}
    U(t) \sim \mathcal{N} \left( 0, k\gamma \right),
\end{align}
which gives $\mu_2$ and $\mathsfi{V}_{22}$.
Using~\eqref{Langevin_trans},~\eqref{U_stat} and lemma~\ref{lemma1} yields $\mu_1$ and $\mathsfi{V}_{11}$.
Finally, $\mathsfi{V}_{12}=\mathsfi{V}_{21}$ can be calculated using~\eqref{Langevin_trans} and the law of total covariance -- specifically,
\begin{subequations}
\begin{align}
    \text{Cov}(\Delta X_n, U_{n+1}) &= \mathbb{E}\left[\text{Cov}(\Delta X_n, U_{n+1}|U_n)\right] + \text{Cov}\left(\mathbb{E}\left[\Delta X_n|U_n\right],\,\mathbb{E}\left[U_{n+1}|U_n\right]\right)\\
    &= \mathsfi{C}_{12} + \text{Cov}\left(U_n\varphi(\gamma\tau)\tau,\, U_n \e^{-\gamma\tau}\right)\\
    &= \mathsfi{C}_{12} + \varphi(\gamma\tau)\tau\e^{-\gamma\tau}\text{Var}(U_n)\\
    &= k\varphi(\gamma\tau)\gamma\tau,
\end{align}
\end{subequations}
recalling that $\text{Var}(U_n)=k\gamma$.

The autocovariance of $\bm{Y}_n$ is defined as
\begin{align}
    \mathsfbi{G}(m) := \mathbb{E}[(\bm{Y}_n-\bm{\mu})(\bm{Y}_{n-m}-\bm{\mu})^{\mathrm{T}}],
\end{align}
where $m\in\mathbb{Z}$. Notice that $\mathsfbi{G}(0)$ is the stationary variance of $\bm{Y}_n$. 
Postmultiplying~\eqref{Langevin_VAR} by $\bm{Y}_{n-m}^{\mathrm{T}}$ and taking expectations gives
\begin{align}
\mathbb{E}\left[
    \bm{Y}_n
    \bm{Y}_{n-m}^{\mathrm{T}}
    \right]
    =
    \mathsfbi{A}\,
\mathbb{E}\left[
    \bm{Y}_{n-1}
    \bm{Y}_{n-m}^{\mathrm{T}}
    \right]
    + 
\mathbb{E}\left[
    \bm{\varepsilon}_n
    \bm{Y}_{n-m}^{\mathrm{T}}
    \right].
\end{align}
Thus, for $m>0$, since $\bm{Y}_{n-m}$ is independent of $\bm{\varepsilon}_n$,
\begin{align}\label{recurrence}
    \mathsfbi{G}(m) = \mathsfbi{A}\, \mathsfbi{G}(m-1)
\end{align}
Therefore, $\mathsfbi{G}(m)$ can be calculated recursively for $m>0$ as
\begin{subequations}
\begin{align}\label{recurrence_sol}
    \mathsfbi{G}(m) &= \mathsfbi{A}^m\, \mathsfbi{G}(0)\\
    &= \mathsfbi{A}^m\, \mathsfbi{V}.
\end{align}
\end{subequations}
Note that~\eqref{recurrence} is an instance of a Yule--Walker equation~\citep[pp. 26-27]{lutkepohl}.

Thus, the joint distribution of a sequence of observations $\{\bm{Y}_n:n\in\{0, \cdots, N_{\tau} - 1\}\}$ is given by 
\begin{align}\label{Y}
    \begin{pmatrix}
    \bm{Y}_0\\
    \bm{Y}_1\\
    \vdots\\
    \bm{Y}_{N_{\tau}-1}
    \end{pmatrix}
    \sim\mathcal{N}\left(\bm{0},
    \begin{pmatrix}
    \mathsfbi{G}(0)&\mathsfbi{G}(1)&\cdots&\mathsfbi{G}(N_{\tau}-1)\\
    \mathsfbi{G}(1) & \ddots & \ddots & \vdots\\
    \vdots & \ddots\\
    &&&\mathsfbi{G}(1)\\
    \mathsfbi{G}(N_{\tau}-1) & & \mathsfbi{G}(1) & \mathsfbi{G}(0)
    \end{pmatrix}\right).
\end{align}

Marginalising~\eqref{Y} for the distribution of $(\Delta X_0, \cdots, \Delta X_{N_{\tau}-1})^{\mathrm{T}}$ we find
\begin{align}\label{DX}
    \begin{pmatrix}
    \Delta X_0\\
    \Delta X_1\\
    \vdots\\
    \Delta X_{N_{\tau}-1}
    \end{pmatrix}
    \sim\mathcal{N}\left(\bm{0},
    \begin{pmatrix}
    \mathsfi{G}_{11}(0)&\mathsfi{G}_{11}(1)&\cdots&\mathsfi{G}_{11}(N_{\tau}-1)\\
    \mathsfi{G}_{11}(1) & \ddots & \ddots & \vdots\\
    \vdots & \ddots\\
    &&&\mathsfi{G}_{11}(1)\\
    \mathsfi{G}_{11}(N_{\tau}-1) & & \mathsfi{G}_{11}(1) & \mathsfi{G}_{11}(0)
    \end{pmatrix}\right).
\end{align}
Using~\eqref{G0} and~\eqref{recurrence} it is easy to see that for $m\geq1$
\begin{align}
    \mathsfbi{G}(m) = 
    \begin{pmatrix}
    k\gamma\tau^2\varphi^2(\gamma\tau)\e^{-(m-1)\gamma\tau}
    & k\gamma\tau\varphi(\gamma\tau)\e^{-(m-1)\gamma\tau}\\
    k\gamma\tau\varphi(\gamma\tau)\e^{-m\gamma\tau}
    & k\gamma \e^{-m\gamma\tau}
    \end{pmatrix}.
\end{align}
Hence, in particular,
\begin{align}\label{G11m}
    \mathsfi{G}_{11}(m) = k\gamma\tau^2\varphi^2(\gamma\tau)\e^{-(m-1)\gamma\tau}.
\end{align}
The likelihood $p(\Delta X_0, \cdots, \Delta X_{N_{\tau}-1})$ is determined by~\eqref{DX} and~\eqref{G11m}.

\section{}
\begin{lemma}
\label{lemma1}
If $X\sim\mathcal{N}(\mu,\,\sigma^2)$ and $Y|X\sim\mathcal{N}(aX,\,\tau^2)$, then $Y\sim\mathcal{N}(a\mu,\,a^2\sigma^2 +\tau^2)$.
\end{lemma}

\begin{proof}
Consider the moment generating function of $Y$, $M_Y(t)$. Recall that for a Gaussian random variable such as $X$ the moment generating function is
\begin{align}
    M_X(t) := \mathbb{E}_X\left[\e^{Xt}\right] = \e^{\mu t+\sigma^2 t^2 / 2};
\end{align}
similarly, since $aX\sim \mathcal{N}\left(a\mu,a^2\sigma^2\right)$,
\begin{align}
    M_{aX}(t) := \mathbb{E}_X\left[\e^{aXt}\right] = \e^{a\mu t+a^2\sigma^2 t^2 / 2}.
\end{align}

Now,
\begin{subequations}
\begin{align}
    M_Y(t) &= \mathbb{E}_Y\left[\e^{Yt}\right]\\
    &= \mathbb{E}_X\left[\mathbb{E}_{Y|X}\left[\e^{Yt}\right]\right]\\
    &= \mathbb{E}_X\left[\e^{aXt + \tau^2t^2/2}\right]\\
    &= \e^{\tau^2t^2/2}\mathbb{E}_X\left[\e^{aXt}\right] = \e^{(a\mu) t+(a^2\sigma^2+\tau^2)t^2/2},
\end{align}
\end{subequations}
which we can recognise as the moment generating function of a Gaussian random variable with mean $a\mu$ and variance $a^2\sigma^2+\tau^2$.
\end{proof}

\bibliographystyle{plainnat}
\bibliography{main}

\end{document}